\documentclass[a4paper,12pt]{article}
\usepackage[totalwidth=500pt,totalheight=680pt]{geometry}
\usepackage{graphicx}
\usepackage{color}
\usepackage{amsmath}
\usepackage{amssymb}
\usepackage{mathrsfs}         
\usepackage[all]{xy}

\usepackage{amsthm}

\usepackage{url}

\newtheorem{lemma}{Lemma} 
\newtheorem{theorem}{Theorem} 
\newtheorem*{theorem*}{Theorem} 
\newtheorem{corollary}{Corollary} 
\newtheorem{proposition}{Proposition}

\theoremstyle{remark}

\newcommand{\comments}[1]{}

\newcommand{\surhom}{
  \ensuremath{
      \negthinspace 
      \longrightarrow
      \hspace{-5mm} \rightarrow \hspace{1mm}
  }
}

\renewcommand{\phi}{\varphi}

\newcommand{\ignore}[1]{}

\title{QCSP on partially reflexive forests}
\author{
Barnaby Martin \\ 
Engineering and Computing Sciences, Durham University, U.K.\thanks{Supported by EPSRC grant EP/G020604/1.} \\
\texttt{barnabymartin@gmail.com}
}
\begin{document}
\maketitle

\begin{abstract}
We study the (non-uniform) quantified constraint satisfaction problem QCSP$(\mathcal{H})$ as $\mathcal{H}$ ranges over partially reflexive forests. We obtain a complexity-theoretic dichotomy: QCSP$(\mathcal{H})$ is either in NL or is NP-hard. The separating condition is related firstly to connectivity, and thereafter to accessibility from all vertices of $\mathcal{H}$ to connected reflexive subgraphs. In the case of partially reflexive paths, we give a refinement of our dichotomy: QCSP$(\mathcal{H})$ is either in NL or is Pspace-complete.
\end{abstract}

\section{Introduction}

The \emph{quantified constraint satisfaction problem} QCSP$(\mathcal{B})$, for a fixed \emph{template} (structure) $\mathcal{B}$, is a popular generalisation of the \emph{constraint satisfaction problem} CSP$(\mathcal{B})$. In the latter, one asks if a primitive positive sentence (the existential quantification of a conjunction of atoms) $\phi$ is true on $\mathcal{B}$, while in the former this sentence may be positive Horn (where universal quantification is also permitted). Much of the theoretical research into CSPs is in respect of a large complexity classification project -- it is conjectured that CSP$(\mathcal{B})$ is always either in P or NP-complete \cite{FederVardi}. This \emph{dichotomy} conjecture remains unsettled, although dichotomy is now known on substantial classes (e.g. structures of size $\leq 3$ \cite{Schaefer,Bulatov} and smooth digraphs \cite{HellNesetril,barto:1782}). Various methods, combinatorial (graph-theoretic), logical and universal-algebraic have been brought to bear on this classification project, with many remarkable consequences. A conjectured delineation for the dichotomy was given in the algebraic language in \cite{JBK}.

Complexity classifications for QCSPs appear to be harder than for CSPs. Just as CSP$(\mathcal{B})$ is always in NP, so QCSP$(\mathcal{B})$ is always in Pspace. However, no overarching polychotomy has been conjectured for the complexities of QCSP$(\mathcal{B})$, as $\mathcal{B}$ ranges over finite structures, but the only known complexities are P, NP-complete and Pspace-complete (see \cite{BBCJK,CiE2006} for some trichotomies). It seems plausible that these complexities are the only ones that can be so obtained. 

In this paper we study the complexity of QCSP$(\mathcal{H})$, where $\mathcal{H}$ is a partially reflexive (undirected) forest, i.e. a forest with potentially some loops. CSP$(\mathcal{H})$, in these instances, will either be equivalent to \emph{$2$-colourability} and be in L (if $\mathcal{H}$ is irreflexive and contains an edge) or will be trivial (if $\mathcal{H}$ contains no edges or some self-loop). Thus, CSP$(\mathcal{H})$ is here always (very) easy. We will discover, however that QCSP$(\mathcal{H})$ may be either in NL or be NP-hard (and is often Pspace-complete).

It is well-known that CSP$(\mathcal{B})$ is equivalent to the \emph{homomorphism problem} \textsc{Hom}$(\mathcal{B})$ -- is there a homomorphism from an input structure $\mathcal{A}$ to $\mathcal{B}$? A similar problem, \textsc{Sur-Hom}$(\mathcal{B})$, requires that this homomorphism be surjective. On Boolean $\mathcal{B}$, each of CSP$(\mathcal{B})$, \textsc{Sur-Hom}$(\mathcal{B})$ and QCSP$(\mathcal{B})$ display complexity-theoretic dichotomy (the first two between P and NP-complete, the last between P and Pspace-complete). However, the position of the dichotomy is the same for QCSP and \textsc{Sur-Hom}, while it is different for CSP. Indeed, the QCSP and \textsc{Sur-Hom} are cousins: a surjective homomorphism from $\mathcal{A}$ to $\mathcal{B}$ is equivalent to a sentence $\Theta$ of the form $\exists v_1,\ldots,v_k \theta(v_1,\ldots,v_k) \wedge \forall y (y=v_1 \vee \ldots \vee y=v_k)$, for $\theta$ a conjunction of atoms, being true on $\mathcal{B}$. This sentence is certainly not positive Horn (it involves some disjunction), but some similarity is there. Recently, a complexity classification for \textsc{Sur-Hom}$(\mathcal{H})$, where $\mathcal{H}$ is a partially reflexive forest, was given in \cite{GolovachPaulusmaSong}.\footnote{Their paper is in fact about partially reflexive trees, but they state in the conclusion how their result extends to partially reflexive forests.} The separation between those cases that are in P and those cases that are NP-complete is relatively simple, those that are hard are precisely those in which, in some connected component (tree), the loops induce a disconnected subgraph. Their work is our principle motivation, but our dichotomy appears more complicated than theirs. Even in the basic case of partially reflexive paths, we find examples $\mathcal{P}$ whose loops induce a disconnected subgraph and yet QCSP$(\mathcal{P})$ is in NL. In the world of QCSP, for templates that are partially reflexive forests $\mathcal{H}$, the condition for tractability may be read as follows. If $\mathcal{H}$ is disconnected (not a tree) then QCSP$(\mathcal{H})$ is in NL. Otherwise, let $\lambda_H$ be the longest distance from a vertex in $H$ to a loop in $\mathcal{H}$. If either 1.) there exists no looped vertex or 2.) there exists a single reflexive connected subgraph $\mathcal{T}_0 \subseteq \mathcal{H}$, such that there is a $\lambda_H$-walk from any vertex of $H$ to $T_0$, then QCSP$(\mathcal{H})$ is in NL (we term such an $\mathcal{H}$ \emph{quasi-loop-connected}). In all other cases, QCSP$(\mathcal{H})$ is NP-hard. In the case of partially reflexive paths, we may go further and state that all other cases are Pspace-complete.

In the world of partially reflexive trees, we derive our NL membership results through the algebraic device of polymorphisms, together with a logico-combinatorial characterisation of template equivalence given in \cite{LICS2008}. In the first instance, we consider trees in which the loops induce a connected subgraph: so-called \emph{loop-connected} trees -- including irreflexive trees. Such trees $\mathcal{T}$ are shown to possess certain (surjective) polymorphisms, that are known to collapse the complexity of QCSP$(\mathcal{T})$ to a polynomially-sized ensemble of instances of CSP$(\mathcal{T}^c)$ (the superscript suggesting an expansions by some constants) \cite{hubie-sicomp}. Although CSP$(\mathcal{T}^c)$ may no longer trivial, $\mathcal{T}^c$ still admits a majority polymorphism, so it follows that CSP$(\mathcal{T}^c)$ is in NL \cite{DalmauK08}.

We prove that every loop-connected tree $\mathcal{T}$ admits a certain \emph{majority} polymorphism, and deduce therefore that QCSP$(\mathcal{T})$ is in NL. However, we also prove that loop-connected trees are the only trees that admit majority polymorphisms, and so we can take this method no further. In order to derive the remaining tractability results, we use the characterisation from \cite{LICS2008} for equivalent templates -- the first time this method has been used in complexity classification. If there exist natural numbers $t$ and $s$ such that there are surjective homomorphisms from $\mathcal{T}^t$ to $\mathcal{S}$ and from $\mathcal{S}^s$ to $\mathcal{T}$ (the superscript here indicates direct power), then it follows that QCSP$(\mathcal{T})$ = QCSP$(\mathcal{S})$, i.e. $\mathcal{T}$ and $\mathcal{S}$ agree on all positive Horn sentences. Of course it follows immediately that QCSP$(\mathcal{T})$ and QCSP$(\mathcal{S})$ are of the same complexity. It turns out that for every quasi-loop-connected tree $\mathcal{T}$, there is a loop-connected tree $\mathcal{S}$ such that QCSP$(\mathcal{T})$ = QCSP$(\mathcal{S})$, and our tractability classification follows (indeed, one may even insist that the loops of $\mathcal{S}$ are always contiguous with some leaves). 

For our NP-hardness proofs we use a direct reduction from \emph{not-all-equal $3$-satisfiability} (NAE3SAT), borrowing heavily from \cite{MatchingCut}. (In the paper \cite{GolovachPaulusmaSong} the NP-hardness results follow by reduction from the problem \emph{matching cut}, which is proven NP-complete in \cite{MatchingCut} by reduction from NAE3SAT.) Our Pspace-hardness proofs, for partially reflexive paths only, use a direct reduction from \emph{quantified not-all-equal $3$-satisfiability} (QNAE3SAT). We require several different flavours of the same reduction in order to cover each of the outstanding cases. We conjecture that all NP-hardness cases for partially reflexive trees (forests) are in fact Pspace-complete.

The paper is organised as follows. After the preliminaries and definitions, we give the cases that are in NL in Section~\ref{sec:easy}, and the cases that are NP-hard and Pspace-complete in Section~\ref{sec:hard}. For the cases that are in NL, we first give our result for loop-connected trees. We then expand to the case of quasi-loop-connected paths (for pedagogy and as an important special subclass) before going on to all quasi-loop-connected trees. For the cases that are hard, we begin with the Pspace-completeness results for paths and then give the NP-hardness for the outstanding trees. Finally we conclude with open problems. We give here our main results.
\begin{theorem}[Pspace Dichotomy]
\label{thm:paths-main}
Suppose $\mathcal{P}$ is a partially reflexive path. Then, either $\mathcal{P}$ is quasi-loop-connected, and QCSP$(\mathcal{P})$ is in NL, or QCSP$(\mathcal{P})$ is Pspace-complete.
\end{theorem}
\begin{proof}
This follows immediately from Theorems~\ref{thm:paths-easy} (tractability) and \ref{thm:paths-hard} (Pspace-completeness).
\end{proof}
\begin{theorem}[NP Dichotomy]
\label{thm:main}
Suppose $\mathcal{H}$ is a partially reflexive forest. Then, either $\mathcal{H}$ is disconnected or quasi-loop-connected, and QCSP$(\mathcal{H})$ is in NL, or QCSP$(\mathcal{H})$ is NP-hard.
\end{theorem}
\begin{proof}
For tractability, if $\mathcal{H}$ is a tree then we appeal to Corollary~\ref{cor:quasi}. If $\mathcal{H}$ is a forest that is not a tree, then it follows that $\mathcal{H}$ is disconnected, and that that QCSP$(\mathcal{H})$ is equivalent to QCSP$(\mathcal{H})$ with inputs restricted to the conjunction of sentences of the form ``$\forall x \exists \overline{y} \phi(x,\overline{y})$'', where $\phi$ is a conjunction of positive atoms (see \cite{CiE2006}). The evaluation of such sentences on any partially reflexive forest is readily seen to be in NL.

For NP-hardness, we appeal to Theorem~\ref{thm:not-quasi-loop-connected}.
\end{proof}

\section{Preliminaries and definitions}

Let $[n]:=\{1,\ldots,n\}$. A graph $\mathcal{G}$ has vertex set $G$, of cardinality $|G|$, and edge set $E(\mathcal{G})$.
Henceforth we consider partially reflexive trees, \mbox{i.e.} trees potentially with some loops (we will now drop the preface partially reflexive).
For a sequence $\alpha \in \{0,1\}^*$, of length $|\alpha|$, let $\mathcal{P}_\alpha$ be the undirected path on $|\alpha|$ vertices such that the $i$th vertex has a loop iff the $i$th entry of $\alpha$ is $1$ (we may say that the path $\mathcal{P}$ is \emph{of the form} $\alpha$). We will usually envisage the domain of a path with $n$ vertices to be $[n]$, where the vertices appear in the natural order. The \emph{centre} of a path is either the middle vertex, if there is an odd number of vertices, or between the two middle vertices, otherwise. Therefore the position of the centre of a path on $m$ vertices is at $\frac{m+1}{2}$. In a path on an even number of vertices, we may refer to the pair of vertices in the middle as \emph{centre vertices}. Call a path $\mathcal{P}$ \emph{loop-connected} if the loops induce a connected subgraph of $\mathcal{P}$.
Call a path \emph{$0$-eccentric} if it is of the form $\alpha 1^b0^a$ for $b\geq 0$ and $|\alpha|\leq a$.
Call a path \emph{weakly balanced} if, proceeding from the centre to each end, one encounters at some point a non-loop followed by a loop (if the centre is loopless then this may count as a non-loop for both directions). Call a weakly-balanced path $\mathcal{P}$ \emph{$0$-centred} if the centre vertex is a non-loop (and $|P|$ is odd) or one of the centre vertices is a non-loop (and $|P|$ is even). Otherwise, a weakly-balanced path $\mathcal{P}$ is \emph{$1$-centred}. 

In a rooted tree, the \emph{height} of the tree is the maximal distance from any vertex to the root. For a tree $\mathcal{T}$ and vertex $v \in T$, let $\lambda_T(v)$ be the shortest distance in $\mathcal{T}$ from $v$ to a looped vertex (if $\mathcal{T}$ is irreflexive, then $\lambda_T(v)$ is always infinite). Let $\lambda_T$ be the maximum of $\{\lambda_T(v):v \in T\}$. A tree is \emph{loop-connected} if the self-loops induce a connected subtree. A tree $\mathcal{T}$ is \emph{quasi-loop-connected} if either 1.) it is irreflexive, or 2.) there exists a connected reflexive subtree $\mathcal{T}_0$ (chosen to be \textbf{maximal} under inclusion) such that there is a walk of length $\lambda_T$ from every vertex of $\mathcal{T}$ to $T_0$. The quasi-loop-connected paths are precisely those that are $0$-eccentric.

The problems CSP$(\mathcal{T})$ and QCSP$(\mathcal{T})$ each take as input a sentence $\phi$, and ask whether this sentence is true on $\mathcal{T}$. For the former, the sentence involves the existential quantification of a conjunction of atoms -- \emph{primitive positive} logic. For the latter, the sentence involves the arbitrary quantification of a conjuction of atoms -- \emph{positive Horn} logic.

The \emph{direct product} $\mathcal{G} \times \mathcal{H}$ of two graphs $\mathcal{G}$ and $\mathcal{H}$ has vertex set $\{(x,y):x \in G, y \in H\}$ and edge set $\{((x,u),(y,v)):x,y \in G, u,v \in H, (x,y) \in E(\mathcal{G}), (u,v) \in E(\mathcal{H})\}$. Direct products are (up to isomorphism) associative and commutative. The $k$th power $\mathcal{G}^k$ of a graph $\mathcal{G}$ is $\mathcal{G} \times \ldots \times \mathcal{G}$ ($k$ times). A homomorphism from a graph $\mathcal{G}$ to a graph $\mathcal{H}$ is a function $h:G\rightarrow H$ such that, if $(x,y) \in E(\mathcal{G})$, then $(h(x),h(y)) \in E(\mathcal{G})$ (we sometimes use ${\surhom}$ to indicate existence of surjective homomorphism). A \emph{$k$-ary polymorphism} of a graph is a homomorphism from $\mathcal{G}^k$ to $\mathcal{G}$. A ternary function $f:G^3 \rightarrow G$ is designated a \emph{majority} operation if $f(x,x,y)=f(x,y,x)=f(y,x,x)=x$, for all $x,y \in G$.

In a matrix, we refer to the \emph{leading} diagonal, running from the top left to bottom right corner, and the \emph{rising} diagonal running from the bottom left to top right corner.

The computational reductions we use will always be comprised by local substitutions that can easily be seen to be possible in logspace -- we will not mention this again. Likewise, we recall that QCSP$(\mathcal{T})$ is always in Pspace, thus Pspace-completeness proofs will only deal with Pspace-hardness.

\section{Tractable trees}
\label{sec:easy}

We now explore the tree templates $\mathcal{T}$ such that QCSP$(\mathcal{T})$ is in NL. We derive our tractability results though majority polymorphisms and equivalence of template.

\subsection{Loop-connected trees and majority polymorphisms}

\subsubsection{Majority operations on (antireflexive) trees}

It is known that all (irreflexive) trees admit a majority polymorphism \cite{Bandelt1987191}; however, not just any operation will suffice for our purposes, therefore we define a majority polymorphism of a certain kind on a rooted tree $\mathcal{T}$ \textbf{whose root could also be a leaf} (\mbox{i.e.} is of degree one). In a rooted tree let the root be labelled $0$ and let the parity propagate outwards from the root along the branches. For $x,y \in T$ define the $\mathrm{meet}(x,y)$ to be the highest (first) point at which the paths from the root to $x$ and the root to $y$ meet. If $x$ and $y$ are on the same branch, and the closer to the root is $x$, then $\mathrm{meet}(x,y)$ is $x$. In the following definition, we sometimes write, e.g., $d$ $[-1]$, to indicate that the function takes either value $d$ or $d-1$: this is dependent on the dominant parity of the arguments which should be matched by the function. Define the following ternary function $f$ on $T$.
\[ f_0(x,y,z):= \left\{
\begin{array}{llr}
d \ [-1] & \mbox{$x,y,z$ all the same parity; $d$ is highest of}   \\
& \mbox{$\mathrm{meet}(x,y)$, $\mathrm{meet}(y,z)$ and $\mathrm{meet}(x,z)$} & A \\
\mbox{meet}(u,v) \ [-1] & \mbox{two of $x,y,z$ ($u$ and $v$) same parity; other different} & B \\
\end{array}
\right.
\]
\begin{lemma}
\label{lem:anti-op}
Let $\mathcal{T}$ be a rooted (irreflexive) tree whose root has degree one. Then $f_0$ is a majority polymorphism of $\mathcal{T}$.
\end{lemma}
\begin{proof}
It is easy to see that $f$ is a majority operation and that it is well-defined. For the latter we use both the fact that the root has degree one and that $\mathrm{meet}(x,y)$, $\mathrm{meet}(y,z)$ and $\mathrm{meet}(x,z)$ are never incomparable. Now we prove that $f$ preserves $E(\mathcal{T})$. Consider $f(a,b,c)$ and $f(a',b',c')$ such that $(a,a')$,  $(b,b')$ and $(c,c') \in E(\mathcal{T})$.

Case I. $a$, $b$ and $c$ are of the same parity. It follows that $a'$, $b'$ and $c'$ are the same parity, and Rule A applies.

Case II. Two are of one parity, one is different -- w.l.o.g. $a$ and $b$ share same parity. Thus, $a'$ and $b'$ are of the same parity with $c'$ different, and Rule B applies.
\end{proof}

\subsubsection{Majority operations on reflexive trees}

It is known that reflexive trees admit a majority polymorphism \cite{Bandelt:1989:DAR:72175.72177}, but it will be a simple matter for us to provide our own. If $x,y,z$ are vertices of a (not necessarily reflexive) tree $\mathcal{T}$ then we define their \emph{median} to be the unique point where the paths from $x$ to $y$, $y$ to $z$ and $x$ to $z$ meet. It follows that median is a majority operation. If $x$, $y$ and $z$ are all on a single branch (path), then we have given the standard definition of median. On a tree, the median function need not be conservative (\mbox{i.e.} we do not in general have $\mathrm{median}(x,y,z)\in \{x,y,z\}$). The following is easy to verify.
\begin{lemma}
Let $\mathcal{T}$ be a reflexive tree. Then the median function is a majority polymorphism of $\mathcal{T}$.
\end{lemma}

\subsubsection{Amalgamating these operations}

Let $\mathcal{T}$ be constructed by attaching rooted (irreflexive) trees -- called \emph{tree-components} -- whose roots have degree one, to the branches of some reflexive tree -- the \emph{centre} -- such that the resulting object is a partially reflexive tree (loop-connected, of course). The roots maintain their labels $0$ despite now having a loop there. These special looped vertices are considered both part of their tree-component(s) and part of the centre. For the sake of well-definition, we preferentially see a vertex $0$ as being in some tree-component. Thus a looped vertex $0$ and the vertex $1$ above it, in its tree-component, constitute two vertices in the same tree-component. It is possible that a looped vertex $0$ is simultaneously the $0$ in multiple tree-components. This will mean we have to verify well-definition. Define the following ternary operation on $T$.
\[ f_1(x,y,z):= \left\{
\begin{array}{llr}
f_0(x,y,z) & \mbox{$x,y,z$ in same tree-component} & A \\
\mathrm{meet}(u,v) \ [-1] & \mbox{two of $x,y,z$  ($u$ and $v$) in the same tree-component;} \\
& \mbox{other elsewhere; and $u, v$ same parity} & B \\
0 \mbox{ from $\{u, v\}$'s comp.} & \mbox{two of $x,y,z$  ($u$ and $v$) in the same tree-component;} \\
& \mbox{other elsewhere; and $u, v$ different parity} & C \\
\mathrm{median}(x,y,z) & \mbox{ otherwise}  & D \\
\end{array}
\right.
\]
\begin{lemma}
\label{lem:maj-f1}
Let $\mathcal{T}$ be a loop-connected tree. Then $f_1$ is a majority polymorphism of $\mathcal{T}$.
\end{lemma}
\begin{proof}
Well-definition can be seen by noting two things. Firstly, Rule B always returns an element in $u$ and $v$'s tree-component. Secondly, if a $0$ is involved that is a $0$ in more than one (\mbox{i.e.} two) relevant tree-components, then Rules B and C return the same value.

A loop-connected tree is either irreflexive, in which case we refer to Lemma~\ref{lem:anti-op}, or is exactly of the form given above (constructed by adding irreflexive trees to a reflexive tree). In the latter case, let us consider $f_1(a,b,c)$ and $f_1(a',b',c')$ such that $(a,a')$,  $(b,b')$ and $(c,c') \in E(\mathcal{T})$.

If $a,b,c$ are all in the same tree-component, then either $a',b',c'$ are all in the same tree-component (and we use Rule A), or one of  Rules B or C applies. We argue symmetrically if $a',b',c'$ are all in the same tree-component.

Otherwise, neither $a,b,c$ nor $a',b',c'$ are triples in the same tree-component. Suppose that two elements from $a,b,c$ (w.l.o.g., $a$ and $b$) are from the same tree-component, and $c$ is from a different tree-component. If $f_1(a,b,c)$ is not the $0$ in $a$ and $b$'s tree-component, then $a'$ and $b'$ are in the same tree-component as $a$ and $b$, with the same respective parities, and either of Rules B or C applies.

The remaining cases all involve $f_1(a,b,c)$ mapping to the reflexive centre, and Rules B, C or D apply.
\end{proof}

\begin{proposition}
\label{prop:loop-connected}
If $\mathcal{T}$ is a loop-connected tree, then QCSP$(\mathcal{T})$ is in NL.
\end{proposition}
\begin{proof}
Since $\mathcal{T}$ admits a majority polymorphism, from Lemma~\ref{lem:maj-f1}, it follows from \cite{hubie-sicomp} that QCSP$(\mathcal{T})$ reduces to the verification of a polynomial number of instances of CSP$(\mathcal{T}^c)$, each of which is in NL by \cite{DalmauK08}. The result follows.
\end{proof}
In fact, for our later results, we only require the majority polymorphism on trees all of whose loops are in a connected component involving leaves. However, we give the fuller result because it is not much more difficult and because we can show these are the only trees admitting a majority polymorphism.
\begin{proposition}
Let $\mathcal{T}$ be a tree that is not loop-connected, then $\mathcal{T}$ does not admit a majority polymorphism. 
\end{proposition}
\begin{proof}
Let $\mathcal{T}$ be such a tree; it follows that it contains an induced path $\mathcal{P}$ of the form $\alpha 1 0^c 1 \beta$. Suppose $\mathcal{P}$ were such a path on $n$ vertices, with $|\alpha|=a$, and $f$ is a majority operation on $P \subseteq T$ with $P:=[n]$ (of course the range of $f$ might include elements in $T \setminus P$). Suppose that $f$ would be a polymorphism of $\mathcal{P}\subseteq \mathcal{T}$ (we will derive a contradiction). Let $p:=a+1 < q:=a+b+1$ be the positions of the loops at the ends of an irreflexive path. Since $f(p,p,q)=p$, we deduce that $f(p,p+1,q)=p$, $p+1$ or any other vertex adjacent to $p$ in $\mathcal{T}$. The same applies to $f(p+1,p,q)$, and it follows that both of $f(p,p+1,q)$ and $f(p+1,p,q)$ can not be $p+1$. W.l.o.g. let $f(p,p+1,q)=r$, where the distance from $r$ to $q$ in $\mathcal{T}$ is greater than or equal to $p-q$. Deduce in turn that the distance from $f(p,p+1,q) \leq p$, \ldots, $f(p,q-1,q \leq q-2)$ to $q$ is greater than or equal to $p-q$, \ldots, $2$. But $f(p,q,q)=q$, and we have the desired contradiction.
\end{proof}

\subsection{Paths of the form $\alpha 0^a$ where $|\alpha|\leq a+1$}

We will now explore the tractability of paths of the form $\alpha 0^a$, where $|\alpha|\leq a+1$. In the proof of the following lemma we deviate from the normalised domain of $[n]$ for a path on $n$ vertices, for pedagogical reasons that will become clear.
\begin{lemma}
\label{lem:surhom}
There is a surjective homomorphism from ${\mathcal{P}_{10^m}}^2$ to $\mathcal{P}_{0^m10^m}$.
\end{lemma}
\begin{proof}
Let $[a,b]:=\{a,\ldots,b\}$. Let $E(\mathcal{P}_{10^m}):=\{(i,j):i,j \in [0,m], j=i+1\} \cup \{(0,0)\}$. 
Let $\mathcal{P}_{0^m10^m}$ be the undirected $2m$-path (on $2m+1$ vertices) such that the middle vertex has a self-loop but none of the others do. 
Formally, $E(\mathcal{P}_{0^m10^m})=\{(i,j):i,j \in [-m,m], j=i+1\} \cup \{(0,0)\}$. The numbering of the vertices is important in the following proof.
We will envisage ${\mathcal{P}_{10^m}}^2$ as a square $(m+1) \times (m+1)$ matrix whose top left corner is the vertex $(0,0)$ which has the self-loop. The entry in the matrix tells one where in $\mathcal{P}_{0^m10^m}$ the corresponding vertex of ${\mathcal{P}_{10^m}}^2$ is to map. It will then be a straightforward matter to verify that this is a surjective homomorphism. By way of example, we give the matrix for $m+1:=7$ in Figure~\ref{fig:figure2} (for all smaller $m$ one may simply restrict this matrix by removing rows and columns from the bottom right). $0$ is sometimes written as $-0$ for (obvious) aesthetic reasons -- we will later refer to the two parts \emph{plus} and \emph{minus} of the matrix. 
\begin{figure}[ht]
\begin{minipage}[b]{0.5\linewidth}
\centering
\includegraphics[scale=1]{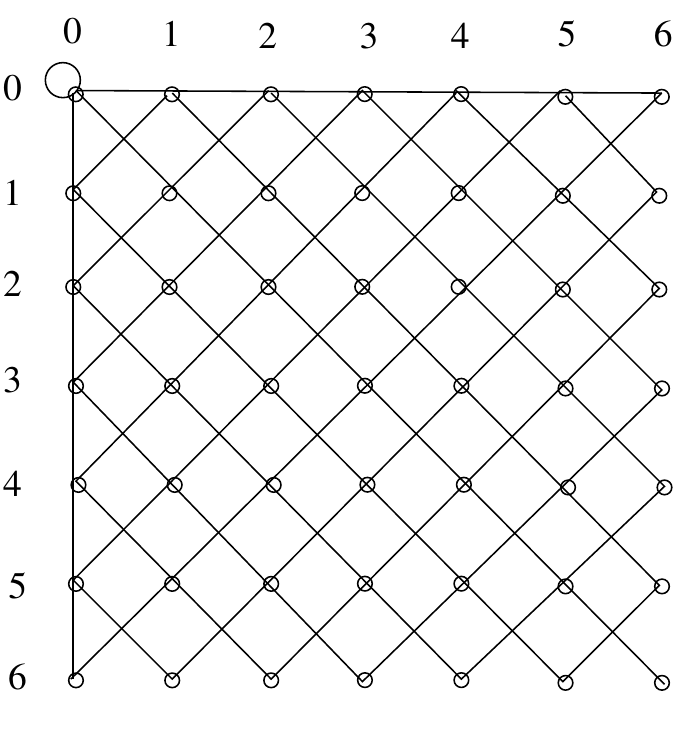}
\caption{${\mathcal{P}_{10^6}}^2$ and its \ldots}
\label{fig:figure1}
\end{minipage}
\hspace{0.5cm}
\begin{minipage}[b]{0.5\linewidth}
\centering
\[
\begin{array}{rrrrrrr}
-0 & 0 &-0 & 0 &-0 & 0 &-0 \\
\\
 1 &-1 & 1 &-1 & 1 &-1 & 1 \\
\\
-0 & 2 &-2 & 2 &-2 & 2 &-2 \\
\\
 1 &-1 & 3 &-3 & 3 &-3 & 3 \\
\\
-0 & 2 &-2 & 4 &-4 & 4 &-4 \\
\\
 1 &-1 & 3 &-3 & 5 &-5 & 5 \\
\\
-0 & 2 &-2 & 4 &-4 & 6 &-6 \\
\end{array}
\]
\caption{\ldots homomorphism to $\mathcal{P}_{0^610^6}$}
\label{fig:figure2}
\end{minipage}
\end{figure}

We refer to the far left-hand column as $0$. Note that the leading diagonal enumerates $-0,...,-m$. Beneath the leading diagonal, the matrix is periodic in each column (with period two). In general, the $j$th column of this matrix will read, from top to bottom:
\[ (-1)^{j-1}.0,\ (-1)^{j}.1,\ (-1)^{j+1}.2,\ \ldots,\ (-1)^{j+j-2}.(j-1),\ -j,\ j+1,\ -j,\ j+1,\ \mbox{etc.} \]
Plainly this map is surjective, we must verify that it is a homomorphism. Adjacent entries in the top row and in the left column must be adjacent -- and this is clearly the case (remember there is a self-loop on $0$). Apart from this, we must have adjacency in the diagonals, thus any point in the matrix must be adjacent (at distance one) from each of its neighbours in the compass directions NW, NE, SE, SW. We consider seven cases and in each of these give the change of value in the matrix for each of the four directions. In the following, LD abbreviates leading diagonal. We begin by considering SW and NW of the leading diagonal (both plusses and minusses), respectively. 

\begin{center}
\begin{tabular}{cccc}
$$
\xymatrix{
+1 \ar@/_0pc/@{<->}[dr] & -1 \ar@/_0pc/@{<->}[dl] \\
+1 & -1
}
$$
&
$$
\xymatrix{
-1 \ar@/_0pc/@{<->}[dr] & +1 \ar@/_0pc/@{<->}[dl] \\
-1 & +1
}
$$
&
$$
\xymatrix{
+1 \ar@/_0pc/@{<->}[dr] & +1 \ar@/_0pc/@{<->}[dl] \\
-1 & -1
}
$$
&
$$
\xymatrix{
-1 \ar@/_0pc/@{<->}[dr] & -1 \ar@/_0pc/@{<->}[dl] \\
+1 & +1
}
$$\\
SW of LD; $-$s & SW of LD; $+$s & NW of LD; $-$s & NW of LD; $+$s 
\end{tabular}
\end{center}

\noindent We now conclude with the three boundary diagonals (the leading diagonal iteself as well as the two diagonals which are adjacent to it).

\begin{center}
\begin{tabular}{cccc}
$$
\xymatrix{
-1 \ar@/_0pc/@{<->}[dr] & -1 \ar@/_0pc/@{<->}[dl] \\
-1 & +1
}
$$
&
$$
\xymatrix{
+1 \ar@/_0pc/@{<->}[dr] & +1 \ar@/_0pc/@{<->}[dl] \\
+1 & -1
}
$$
&
$$
\xymatrix{
-1 \ar@/_0pc/@{<->}[dr] & -1 \ar@/_0pc/@{<->}[dl] \\
+1 & +1
}
$$
\\
Below LD; $+$s & On LD; $-$s & Above LD; $+$s
\end{tabular}
\end{center}
\end{proof}  

\begin{proposition}
\label{prop:0-unbalanced}
If $\mathcal{P}$ is of the form $\alpha 0^a$ where $|\alpha|\leq a+1$, then QCSP$(\mathcal{P})$ is in P.
\end{proposition}
\begin{proof}
Let $\mathcal{P}$ be of the form $\alpha 0^a$ where $|\alpha|\leq a+1$. If $\mathcal{P}$ is irreflexive then the result follows from \cite{CiE2006}. Otherwise, $\mathcal{P}$ contains a loop, the right-most of which is $m$ vertices in from the right-hand end (on or left of centre).
We claim $\mathcal{P} \surhom \mathcal{P}_{10^m}$ and ${\mathcal{P}_{10^m}}^2 \surhom \mathcal{P}$. It then follows from \cite{LICS2008} that QCSP$(\mathcal{P})$ = QCSP$(\mathcal{P}_{10^m})$, whereupon membership in NL follows from Proposition~\ref{prop:loop-connected}.

The surjective homomorphism from $\mathcal{P}$ to $\mathcal{P}_{10^m}$ is trivial: map all vertices to the left of the right-most loop of $\mathcal{P}$ to the loop of $\mathcal{P}_{10^m}$, and let the remainder of the map follows the natural isomorphism. The surjective homomorphism from ${\mathcal{P}_{10^m}}^2 \surhom \mathcal{P}$ follows from the obvious surjective homomorphism from $\mathcal{P}_{0^m10^m}$ to $\mathcal{P}$, via Lemma~\ref{lem:surhom}.
\end{proof}

\subsection{Paths of the form $\alpha 1^b 0^a$ where $b\geq 1$ and $|\alpha|=a$}
\begin{lemma}
\label{lem:surhom2}
For $b\geq 1$, there is a surjective homomorphism from ${\mathcal{P}_{1^a1^b0^a}}^2$ to $\mathcal{P}_{0^a1^b0^a}$.
\end{lemma}
\begin{proof}
We begin by giving the proof for the case $b=1$.
Let $\mathcal{P}_{1^a10^a}$ be the undirected $2a$-path (on $2a+1$ vertices) such that the middle vertex and all vertices to the left have a self-loop but none of the other vertices do. 
Formally, $E(\mathcal{P}_{1^a10^a})=\{(i,j):i,j \in [-a,a], j=i+1\} \cup \{(-a,-a),\ldots,(0,0)\}$. We see $\mathcal{P}_{0^a1^b0^a}$ as a subgraph of $\mathcal{P}_{1^a1^b0^a}$ in the obvious fashion. The numbering of the vertices is important in the following proof.
We will envisage ${\mathcal{P}_{1^a10^a}}^2$ as a square $(2a+1) \times (2a+1)$ matrix whose top left corner is the vertex $(-a,-a)$. The entry in the matrix tells one where in $\mathcal{P}_{0^a10^a}$ the corresponding vertex of ${\mathcal{P}_{1^a10^a}}^2$ is to map. It will then be a straightforward matter to verify that this is a surjective homomorphism. By way of example, we give the matrix for $a:=4$ in Figure~\ref{fig:figure22}.
\begin{figure}[ht]
\begin{minipage}[b]{0.5\linewidth}
\centering
\includegraphics[scale=1]{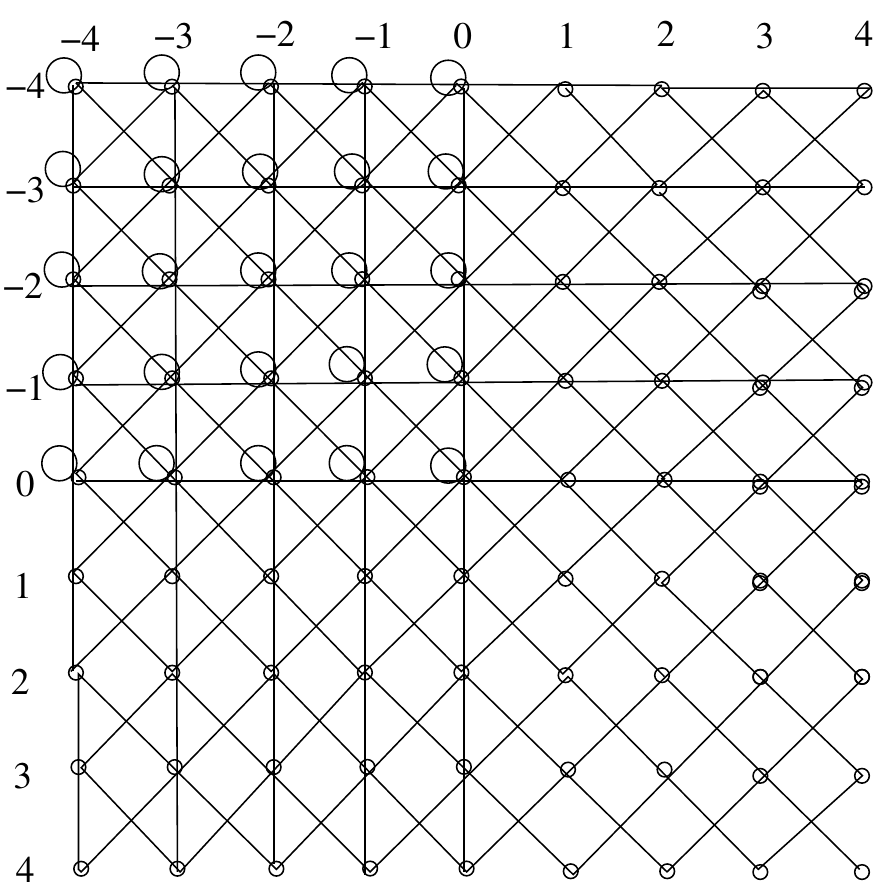}
\caption{${\mathcal{P}_{1^410^4}}^2$ and its \ldots}
\label{fig:figure11}
\end{minipage}
\hspace{0.5cm}
\begin{minipage}[b]{0.5\linewidth}
\centering
\[
\begin{array}{rrrrrrrrr}
 0 & 0 & 0 & 0\ & 0\ & 0\ & 0\ & 0\ & 0 \\
\\
 0 & 0 & 0 & 0\ & 0\ & 0\ & 0\ & 0\ & 0 \\
\\
 0 & 0 & 0 & 0 & 0 & 0 & 0 & 0 & 0 \\
\\
 0 & 0 & 0 & 0 & 0 & 0 & 0 & 0 & 0 \\
\\
 0 & 0 & 0 & 0 & 0 & 0 & 0 & 0 & 0 \\
\\
-1 &-1 &-1 &-1 & 0 & 1 & 1 & 1 & 1 \\
\\
-2 &-2 &-2 & 0 & 0 & 0 & 2 & 2 & 2 \\
\\
-3 &-3 &-1 &-1 & 0 & 1 & 1 & 3 & 3 \\
\\
-4 &-2 &-2 & 0 & 0 & 0 & 2 & 2 & 4 \\
\end{array}
\]
\caption{\ldots homomorphism to $\mathcal{P}_{0^410^4}$}
\label{fig:figure22}
\end{minipage}
\end{figure}
We consider the matrix to have a central cross of $0$s -- indexed by column $0$ and row $0$ -- whose removal leaves four segments. The NW and NE segments are all $0$s and the SW and SE segments are isomorphic under a reflection that maps $-x$ to $x$. We henceforth discuss only the SE submatrix -- columns $\geq 1$ and rows $\geq 1$. Beneath the leading diagonal, the submatrix is periodic in each column (with period two). In general, the $j$th column of the submatrix will read, from top to bottom: $1\, \ldots, j , j-1 ,j ,j-1 , \mbox{etc.}$.

This map is certainly surjective, but we must now verify that it is a homomorphism. In the central cross and NW and NE this is trivial. Where the central cross borders the SW and SE segments this is again easy to see. The argument for within the SW and SE segments is the same. Note that we must have more adjacencies in the SW than the SE -- so we will undertake our argument there. In the SW, we require adjacencies along both diagonals and along the vertical. For the vertical, the necessary relationship, decreasing by one until the rising diagonal, and then oscillating with period two, is immediate. For the diagonals running SW-NE, the values are incrementing by one and the result is clear. For the diagonals running NW-SE, the relationship is slightly more sophisticated, with the values dropping from $-1$ until their minimum, on or immdediately before the rising diagonal, and then climbing again to either $0$ or $-1$. In any case the necessary property holds and the result follows.    

In the case that $b>1$, we will keep the vertices $-a,\ldots,-1$ and $1,\ldots,a$ as the vertices to the left and to the right of the middle loops $0_1,\ldots,0_b$, respectively. Formally, we envisage 
\[
\begin{array}{ll}
E(\mathcal{P}_{1^a1^b0^a}):= & \{(i,j):i,j \in [-a,-1],[1,a], j=i+1\} \cup \{(0_i,0_j):i,j \in [1,b],\} \cup \\ 
& \{(-1,0_1),(0_1,-1),(0_b,1),(1,0_b \} \cup \{(\{(-a,-a),\ldots,(-1,-1),(0_1,0_1),\ldots,(0_b,0_b)\}.
\end{array}
\] 
Now the central cross of the matrix giving the homomorphism from ${\mathcal{P}_{1^a1^b0^a}}^2$ consists of $b$ rows and $b$ columns indexed respectively by $0_1$ to $0_b$. We may assume that the entries of the $i$ such column $0_i$ are uniformly $0_i$, for $i \in [b]$. This accounts for the vertical in the cross. We may now write $0_1$ into the entire left-side of the cross and the NW segment, and $0_b$ into the right-hand side of the cross and the NE segment. The SW and SE segments are specified as before ($0_1$ should be used as $0$ in the SW segment and $0_b$ should be used as $0$ in the SE segment. By way of example, we give the matrix for $a:=5$ and $b:=3$ in Figure~\ref{fig:big-matrix}. The proof of the correctness of this homomorphism is essentially as before (in the proof of Lemma~\ref{lem:surhom}).
\end{proof}
\begin{figure}
\[
\begin{array}{rrrrrrrrrrrrr}
 0_1 & 0_1 & 0_1 & 0_1 & 0_1 & 0_1 & 0_2 & 0_3 & 0_3 & 0_3 & 0_3 & 0_3 & 0_3 \\
\\
 0_1 & 0_1 & 0_1 & 0_1 & 0_1 & 0_1 & 0_2 & 0_3 & 0_3 & 0_3 & 0_3 & 0_3 & 0_3 \\
\\
 0_1 & 0_1 & 0_1 & 0_1 & 0_1 & 0_1 & 0_2 & 0_3 & 0_3 & 0_3 & 0_3 & 0_3 & 0_3 \\
\\
 0_1 & 0_1 & 0_1 & 0_1 & 0_1 & 0_1 & 0_2 & 0_3 & 0_3 & 0_3 & 0_3 & 0_3 & 0_3 \\
\\
 0_1 & 0_1 & 0_1 & 0_1 & 0_1 & 0_1 & 0_2 & 0_3 & 0_3 & 0_3 & 0_3 & 0_3 & 0_3 \\
\\
 0_1 & 0_1 & 0_1 & 0_1 & 0_1 & 0_1 & 0_2 & 0_3 & 0_3 & 0_3 & 0_3 & 0_3 & 0_3 \\
\\
 0_1 & 0_1 & 0_1 & 0_1 & 0_1 & 0_1 & 0_2 & 0_3 & 0_3 & 0_3 & 0_3 & 0_3 & 0_3 \\
\\
 0_1 & 0_1 & 0_1 & 0_1 & 0_1 & 0_1 & 0_2 & 0_3 & 0_3 & 0_3 & 0_3 & 0_3 & 0_3 \\
\\
   -1 &   -1 &   -1 &   -1 &   -1 & 0_1 & 0_2 & 0_3 &   1 &    1 &   1  &    1 &    1  \\
\\
   -2 &   -2 &   -2 &   -2 & 0_1 & 0_1 & 0_2 & 0_3 & 0_3 &    2 &   2  &    2 &    2  \\
\\
   -3 &   -3 &   -3 &   -1 &   -1 & 0_1 & 0_2 & 0_3 &   1 &    1 &   3  &    3 &    3  \\
\\
   -4 &   -4 &   -2 &   -2 & 0_1 & 0_1 & 0_2 & 0_3 & 0_3 &    2 &   2  &    4 &    4  \\
\\
   -5 &   -3 &   -3 &   -1 &   -1 & 0_1 & 0_2 & 0_3 &   1 &    1 &   3  &    3 &    5  \\
\end{array}
\]
\caption{The homomorphism of ${\mathcal{P}_{1^51^3 0^5}}^2$ to $\mathcal{P}_{0^51^3 0^5}$}
\label{fig:big-matrix}
\end{figure}
\begin{proposition}
\label{prop:0-eccentric}
If $\mathcal{P}$ is of the form $\alpha 1^b 0^a$ for $b\geq 1$ and $|\alpha|=a$, then QCSP$(\mathcal{P})$ is in P.
\end{proposition}
\begin{proof}
Let $\mathcal{P}$ be of the form $\alpha 1^b 0^a$ where $|\alpha|=a$ and $b\geq 1$. We claim $\mathcal{P} \surhom \mathcal{P}_{1^a 1^b 0^a}$ and ${\mathcal{P}_{1^a1^b0^a}}^2 \surhom \mathcal{P}$. It then follows from \cite{LICS2008} that QCSP$(\mathcal{P})$ = QCSP$(\mathcal{P}_{1^a1^b0^a})$, whereupon membership in NL follows from Proposition~\ref{prop:loop-connected}.

The surjective homomorphism from $\mathcal{P}$ to $\mathcal{P}_{1^a 1^b 0^a}$ is the identity. The surjective homomorphism from ${\mathcal{P}_{1^a 1^b 0^a}}^2$ to $\mathcal{P}$ follows from the surjective homomorphism from $\mathcal{P}_{0^a1^b0^a}$ to $\mathcal{P}$ (the identity), via Lemma~\ref{lem:surhom2}.
\end{proof}

\begin{theorem}
\label{thm:paths-easy}
If $\mathcal{P}$ is quasi-loop-connected ($0$-eccentric), then QCSP$(\mathcal{P})$ is in P
\end{theorem}
\begin{proof}
If $\mathcal{P}$ is quasi-loop-connected ($0$-eccentric), then either $\mathcal{P}$ is of the form $\alpha 0^a$, for $|\alpha| \leq a+1$, or $\alpha 1^b 0^a$, for $|\alpha|=a$ and $b\geq 1$ (or both!). The result follows from Propositions~\ref{prop:0-unbalanced} and \ref{prop:0-eccentric}.
\end{proof}
We remark that one could prove the tractability results for $0$-eccentric paths using only the method of Lemma~\ref{lem:surhom2}. This is because one can see that there is a surjective homomorphism from ${\mathcal{P}_{1^a10^b}}^2$ to $\mathcal{P}_{0^a10^b}$, when $a<b$ (read this from Figure~\ref{fig:figure22} by removing $b-a$ rows and columns from the top left). However, we will shortly see a generalised version of Lemma~\ref{lem:surhom}, and so it is that our exertions were not in vain.

\subsection{The quasi-loop-connected case}

Suppose that $\mathcal{T}$ is a quasi-loop-connnected tree, that is neither reflexive nor irreflexive, with associated $\mathcal{T}_0$ and $\lambda_T$, as defined in the preliminaries. Let $v_\lambda \in T$ be such that there is no $(\lambda_T-1)$-walk to a looped vertex but there is a $\lambda_T$-walk to the looped vertex $l$ of the maximal (under inclusion) connected reflexive subtree $\mathcal{T}_0$ (such a $v_\lambda$ exists). Let $\mathcal{T}_1$ be the maximal subtree of $\mathcal{T}$ rooted at $l$ that contains $v_\lambda$.
\begin{lemma}
If $\mathcal{T}$ and $v_\lambda$ are as in the previous paragraph, then $v_\lambda$ is a leaf.
\end{lemma}
\begin{proof}
If not, then $v_\lambda$ has a neighbour $w$ on the path in the direction from $l$ towards and beyond $v$. But the distance from this vertex to the connected component $T_0$ containing $l$ is $\lambda_T+1$, which contradicts maximality of $\lambda_T$. 
\end{proof}
\noindent Note that if $\mathcal{T}$ were an arbitrary tree, \mbox{i.e.} not quasi-loop-connected, then there is no need for $v_\lambda$, at maximal distance from a loop, to be a leaf. E.g., let $\mathcal{P}_{101}$ be the path on three vertices, the two ends of which are looped. $\lambda_{P_{101}}=1$ and $v_\lambda$ would be the centre vertex.

So, as before, let $\mathcal{T}$ be a quasi-loop-connnected tree that is neither reflexive nor irreflexive, and let some $v_\lambda \in T$ be given ($v_\lambda$, of course, need not be unique).
There is an irreflexive path $\mathcal{P} \subseteq \mathcal{T}_1$ of length $\lambda_T$ from the leaf $v_\lambda$ to $l \in \mathcal{T}_0$. There may be other paths joining this path, of course. Let $\mathcal{T}'$ be $\mathcal{T}$ with these other paths pruned off (see Figure~\ref{fig:Tdash}). We need to take a short diversion in which we consider graphs with a similar structure to $\mathcal{T}'$.
\begin{figure}[ht]
\begin{minipage}[b]{0.5\linewidth}
\centering
\includegraphics[scale=1]{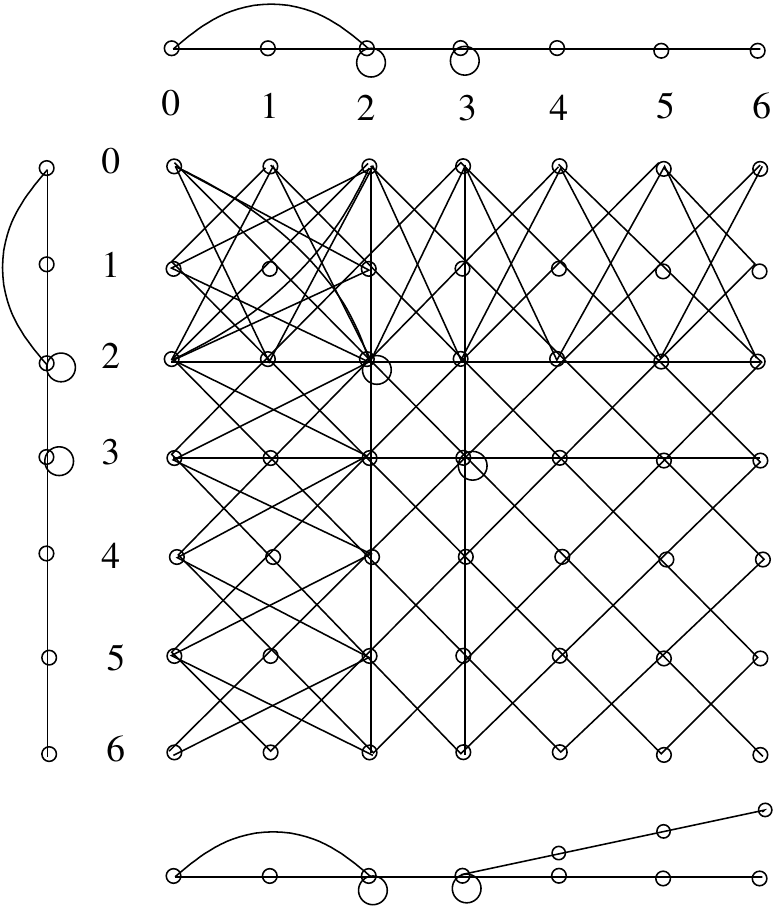}
\caption{$\mathcal{H}^2$ (and $\mathcal{H}'$ below) and its \ldots}
\label{fig:P6squaredSpecial}
\end{minipage}
\hspace{0.5cm}
\begin{minipage}[b]{0.5\linewidth}
\centering
\[
\begin{array}{rrr|rrrr}
 0 & 0 & 0 & 0 & 0 & 0 & 0 \\
& & & & & & \\
 1 & 1 & 1 & 1 & 1 & 1 & 1 \\
& & & & & & \\
 2 & 2 & 2 & 2 & 2 & 2 & 2 \\
\hline
& & & & & & \\
 3 & 3 & 3 & 3 & 3 & 3 & 3 \\
& & & & & & \\
 3 & 3 & 3 & 4 & 4' & 4 & 4' \\
& & & & & & \\
 3 & 3 & 3 & 3 & 5 & 5' & 5 \\
& & & & & & \\
 3 & 3 & 3 & 4 & 4' & 6 & 6' \\
\\
\\
\\
\end{array}
\]
\caption{\ldots homomorphism to $\mathcal{H}'$}
\label{fig:P6squaredSpecial2}
\end{minipage}
\end{figure}
\begin{lemma}
\label{lem:multiply-paths}
Suppose $\mathcal{H}$ consists of a graph $\mathcal{G}$, with a looped vertex $l$, onto which an irreflexive path $\mathcal{P}$ of length $\lambda$ is attached. Let $\mathcal{H}'$ be constructed as $\mathcal{H}$ but with the addition of two (disjoint) paths $\mathcal{P}$ onto the looped vertex $l$. Then there is a surjective homomorphism from $\mathcal{H}^2$ to $\mathcal{H}'$.
\end{lemma}
\begin{proof}
We proceed essentially by example, see Figures~\ref{fig:P6squaredSpecial} and \ref{fig:P6squaredSpecial2}, where $G=\{0,1,2,3\}$, $l=3$ and $\mathcal{P}=\{3,4,5,6\}$. $\{4,5',6',7'\}$ constitutes the copy of $\mathcal{P}$. The generalisation to arbitrary graphs is clear -- we proceed as in the proof of Lemma~\ref{lem:surhom} (for longer $\mathcal{P}$, in the SE part of the matrix) and by the trivial projection (for larger $\mathcal{G}$, in the N of the matrix). For the SW of ther matrix, we continue to map uniformly to $l$ (note that we are requiring the loop on the vertex $l$). 
\end{proof}
\begin{figure}
\begin{center}
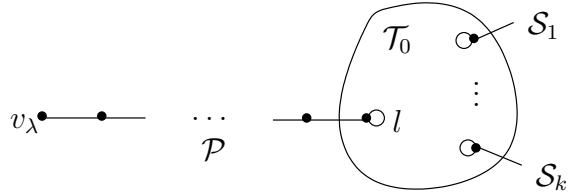
\end{center}
\caption{Anatomy of $\mathcal{T}'$}
\label{fig:Tdash}
\end{figure}
\begin{lemma}
There is a surjective homomorphism from $\mathcal{T}$ to $\mathcal{T}'$. There exists $p\in \mathbb{N}$ such that there is a surjective homomorphism from $(\mathcal{T}')^p$ to $\mathcal{T}$.
\end{lemma}
\begin{proof}
The surjective homomorphism from $\mathcal{T}$ to $\mathcal{T}'$ takes the paths constituted by $\mathcal{T}_1 \setminus \mathcal{P}$ and folds them back towards $l$. These paths may have loops on them, but never at distance $<\lambda_T$ from $v_\lambda$, which explains why this will be a homomorphism.

The surjective homomorphism from $(\mathcal{T}')^p$ to $\mathcal{T}$ comes from the multiplication of the paths $\mathcal{P}$ -- by iteration of Lemma~\ref{lem:multiply-paths} -- in powers of $\mathcal{T}'$ (note that nothing may be further than $\lambda_T$ from $l$ in $\mathcal{T}_1$, without violating maximality of $\lambda_T$ or uniqueness of $\mathcal{T}_0$). To cover $\mathcal{T}_1$ in $\mathcal{T}$ we require no more than $|T_1|$ copies of the path $\mathcal{P}$. According to the previous lemma, we may take $p:=|T_1|-1$ (in fact it is easy to see that $\lceil \log |T_1|\rceil$ suffices).
\end{proof}
Now, it may be possible that in $\mathcal{T}'$ there are subtrees $\mathcal{S}_1$, \ldots, $\mathcal{S}_k$ rooted in $\mathcal{T}_0$ whose first vertex, other than their root, is a non-loop (because we chose $\mathcal{T}_0$ to be maximal under inclusion). The height of these trees is $\leq \lambda_T$. Let $\mathcal{T}''$ be $\mathcal{T}'$ with these subtrees $\mathcal{S}_1,\ldots,\mathcal{S}_k$ being reflexively closed. In the following lemma we use implicitly the fact that the presence of a surjective homomorphism $h$ from $\mathcal{A}^2$ to $\mathcal{B}$, generates also a surjective homomorphism from $\mathcal{A}^4$ to $\mathcal{B}^2$ given by $(a_1,a_2,a_3,a_4)\mapsto (h(a_1,a_2),h(a_3,a_4))$.
\begin{lemma}
There is a surjective homomorphism from $\mathcal{T}'$ to $\mathcal{T}''$. There is a $p$ such that there is a surjective homomorphism from $(\mathcal{T}'')^p$ to $\mathcal{T}'$.
\end{lemma}
\begin{proof}
The identity is a surjective homomorphism from $\mathcal{T}'$ to $\mathcal{T}''$. 

The surjective homomorphism from from $(\mathcal{T}'')^p$ to $\mathcal{T}'$ may be constructed in a variety of stages of repeated squaring, dealing with $\mathcal{S}_i$ in the $i$th stage. There are two cases to consider. Either 1.) there is a $\lambda_T$ walk from $l$ to all vertices of $\mathcal{S}_i$ or 2.) there is not. 

In Case 1, we repeatedly square, reproducing the path $\mathcal{P}$ $|S_i|$ times and covering $\mathcal{S}_i$ in this manner.

Case 2 is illustrated in \comments{Figure~\ref{fig:matrix}}Figure~9, and applies the method from the proof of Lemma~\ref{lem:surhom2}. In this case, the centre of a path from $v_\lambda$, via $l$, to the top of the tree $\mathcal{S}_i$, lies in $\mathcal{T}_0$. in \comments{Figure~\ref{fig:matrix}}Figure~9, $\mathcal{X}$ should be read as $\mathcal{T}_0:=\{2,3,4,5,6\}$, $\mathcal{S}:=\{0,1,2,3\}$, $\mathcal{P}:=\{6,7,8\}$ rooted at $l$ ($l:=6$, $v_\lambda:=8$) and $\mathcal{Z}:=\{4,9,10,11\}$. The purpose of $\mathcal{Z}$ in the picture is to demonstrate that the technique for ``removing'' loops from the reflexive closure of $\mathcal{S}$ can be applied across $\mathcal{T}_0$ (from the non-loops, here in the form of $\{7,8\}$) regardless of any subtrees that may come off from $\mathcal{T}_0$. In $\mathcal{X}$ we have the reflexive closure of $\mathcal{S}$, for $\mathcal{Y}$ we show how to surjectively cover $S$, even if $\mathcal{S}$ had no loops in $\mathcal{T}'$ (other than its root $3$). The centre part of the matrix considers how the square of the path substructure of $\mathcal{X}$ induced by $\{1,\ldots,8\}$ maps to the path substructure of $\mathcal{Y}$ induced by $\{1,\ldots,8\}$ -- this is essentially what we have seen in Lemma~\ref{lem:surhom2}. The left-centre part of the matrix again uses Lemma~\ref{lem:surhom2}, but now we are interested in the path substructures induced by $\{0,2,\ldots,8\}$. The whole right-hand (W) part of the matrix is a projection on to $\mathcal{Z}$, and the top-centre, top-left, bottom-centre and bottom-left follow the pattern of the top of the centre part of the matrix (again, as dictated in Lemma~\ref{lem:surhom2}).

It follows, by repeated squaring and surjective homomorphism according to the stages given, that one may take $p:=2^{|S_1|+\ldots+|S_k|}$ (this $p$ is far from optimal).
\end{proof}
\begin{figure}
\label{fig:matrix}
\begin{center}
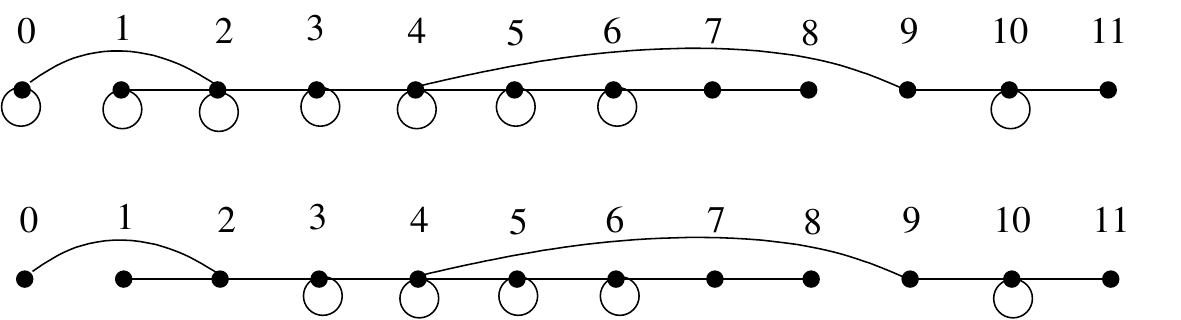
\end{center}
\[
\begin{array}{r|rrrrrrrr|rrr}
 3   & 3   & 3   & 3   & 4   & 5   & 6   & 6   & 6   & 9   & 10  & 11   \\
\hline
 3   & 3   & 3   & 3   & 4   & 5   & 6   & 6   & 6   & 9   & 10  & 11   \\
 3   & 3   & 3   & 3   & 4   & 5   & 6   & 6   & 6   & 9   & 10  & 11   \\
 3   & 3   & 3   & 3   & 4   & 5   & 6   & 6   & 6   & 9   & 10  & 11   \\
 3   & 3   & 3   & 3   & 4   & 5   & 6   & 6   & 6   & 9   & 10  & 11   \\
 3   & 3   & 3   & 3   & 4   & 5   & 6   & 6   & 6   & 9   & 10  & 11   \\
 3   & 3   & 3   & 3   & 4   & 5   & 6   & 6   & 6   & 9   & 10  & 11   \\
 2   & 2   & 2   & 3   & 4   & 5   & 6   & 7   & 7   & 9   & 10  & 11   \\
 0   & 1   & 3   & 3   & 4   & 5   & 6   & 6   & 8   & 9   & 10  & 11   \\
\hline
 3   & 3   & 3   & 3   & 4   & 5   & 6   & 6   & 6   & 9   & 10  & 11   \\
 3   & 3   & 3   & 3   & 4   & 5   & 6   & 6   & 6   & 9   & 10  & 11   \\
 3   & 3   & 3   & 3   & 4   & 5   & 6   & 6   & 6   & 9   & 10  & 11   \\
\end{array}
\]
\caption{Example surjective homomorphism from $\mathcal{X}^2$ to $\mathcal{Y}$.}
\end{figure}
\begin{corollary}
\label{cor:quasi}
Let $\mathcal{T}$ be quasi-loop-connected, then QCSP$(\mathcal{T})$ is in NL.
\end{corollary}
\begin{proof}
If $\mathcal{T}$ is actually loop-connected, then the result is Proposition~\ref{prop:loop-connected}. Otherwise,
QCSP$(\mathcal{T})$= QCSP$(\mathcal{T}')$= QCSP$(\mathcal{T}'')$, and tractability of the last follows from Proposition~\ref{prop:loop-connected}.
\end{proof}

\section{Hard cases}
\label{sec:hard}

\subsection{Pspace-completeness results for paths that are not $0$-eccentric}

\subsubsection{$\mathcal{P}_{101}$ and weakly balanced $0$-centred paths}

In the following proof we introduce the notions of \emph{pattern} and \emph{$\forall$-selector} that will recur in future proofs.
\begin{figure}
\label{fig:gadgets}
\begin{center}
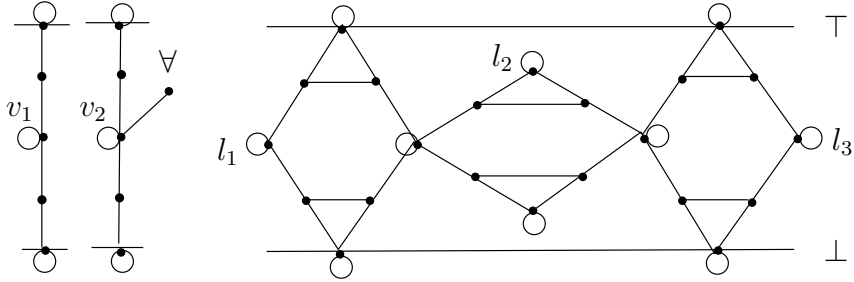
\end{center}
\caption{Variable and clause gadgets in reduction to QCSP$(\mathcal{P}_{101})$.}
\end{figure}
\begin{proposition}
\label{prop:P101}
QCSP$(\mathcal{P}_{101})$ is Pspace-complete.
\end{proposition}
\begin{proof}
For hardness, we reduce from QNAE3SAT, where we will ask for the extra condition that no clause has three universal variables (of course, any such instance would be trivially false). From an instance $\Phi$ of QNAESAT we will build an instance $\Psi$ of QCSP$(\mathcal{P}_{101})$ such that $\Phi$ is in QNAE3SAT iff $\Psi$ in QCSP$(\mathcal{P}_{101})$. We will consider the quantifier-free part of $\Psi$, itself a conjunction of atoms, as a graph, and use the language of homomorphisms. The constraint satisfaction problem, CSP$(\mathcal{P}_{101})$, seen in this guise, is nothing other than the question of homomorphism of this graph to $\mathcal{P}_{101}$. The idea of considering QCSP$(\mathcal{P}_{101})$ as a special type of homomorphism problem is used implicitly in \cite{OxfordQuantifiedConstraints}\footnote{The journal version of this paper was published much later as \cite{BBCJK}.} and explicitly in \cite{CiE2006}.

We begin by describing a graph $\mathcal{G}_\Phi$, whose vertices will give rise to the variables of $\Psi$, and whose edges will give rise to the facts listed in the quantifer-free part of $\Psi$. Most of these variables will be existentially quantified, but a small handful will be universally quantified. $\mathcal{G}_\Phi$ consists of two reflexive paths, labelled $\top$ and $\bot$ which contain inbetween them gadgets for the clauses and variables of $\Phi$. We begin by assuming that the paths $\top$ and $\bot$ are evaluated, under any homomorphism we care to consider, to vertices $1$ and $3$ in $P_{101}$, respectively (the two ends of $P_{101}$); later on we will show how we can effectively enforce this. Of course, once one vertex of one of the paths is evaluated to, say, $1$, then that whole path must also be so evaluated -- as the only looped neighbour of $1$ in $\mathcal{P}_{101}$ is $1$. The gadgets are drawn in \comments{Figure~\ref{fig:gadgets}}Figure~10. The pattern is the path $\mathcal{P}_{101}$, that forms the edges of the diamonds in the clause gadgets as well as the tops and bottoms of the variable gadgets. The diamonds are \emph{braced} by two horizontal edges, one joining the centres of the top patterns and the other joining the centres of the bottom patterns. We will return to the question of the absence of vertical bracing at the end of the proof. The $\forall$-selector is the path $\mathcal{P}_{10}$, which travels between the universal variable node $v_2$ and the labelled vertex $\forall$.

For each existential variable $v_1$ in $\Phi$ we add the gadget on the far left, and for each universal variable $v_2$ we add the gadget immediately to its right. There is a single vertex in that gadget that will eventually give rise to a variable in $\Psi$ that is universally quantified, and it is labelled $\forall$. For each clause of $\Phi$ we introduce a copy of the clause gadget drawn on the right. We then introduce an edge between a variable $v$ and literal $l_i$ ($i \in \{1,2,3\}$) if $v=l_i$ (note that all literals in QNAE3SAT are positive). We reorder the literals in each clause, if necessary, to ensure that literal $l_2$ of any clasue is never a variable in $\Phi$ that is universally quantified. It is not hard to verify that homomorphisms from $\mathcal{G}_\Phi$ to $\mathcal{P}_{101}$ (such that the paths $\top$ and $\bot$ are evaluated to $1$ and $3$, respectively) correspond exactly to satisfying not-all-equal assignments of $\Phi$. The looped vertices must map to either $1$ or $3$ -- $\top$ or $\bot$ -- and the clause gadgets forbid exactly the all-equal assignments. Now we will consider the graph $\mathcal{G}_\Phi$ realised as a formula $\Psi''$, in which we will existentially quantify all of the variables of $\Psi''$ except: one variable each, $v_\top$ and $v_\bot$, corresponding respectively to some vertex from the paths $\top$ and $\bot$; all variables corresponding to the centre vertex of an existential variable gadget; all variables corresponding to the centre vertex of a universal variable gadget, and all variables corresponding to the extra vertex labelled $\forall$ of a universal variable gadget. We now build $\Psi'$ by quantifying, adding outermost and in the order of the quantifiers of $\Phi$:
\begin{itemize}
\item existentially, the variable corresponding to the centre vertex of an existential variable gadget,
\item universally, the variable corresponding to the extra vertex labelled $\forall$ of a universal variable gadget, and then existentially, the variable corresponding to the centre vertex of a universal variable gadget.
\end{itemize}
\noindent The reason we do not directly universally quantify the vertex associated with a universal variable is because we want it to be forced to range over only the looped vertices $1$ and $3$ (which it does as its unlooped neighbour $\forall$ is forced to range over all $\{1,2,3\}$). $\Psi'(v_\top,v_\bot)$ is therefore a positive Horn formula with two free variables, $v_\top$ and $v_\bot$, such that, $\Phi$ is QNAE3SAT iff $\mathcal{P}_{101} \models \Psi'(1,3)$. 
Finally, we construct $\Psi$ from $\Psi'$ with the help of two $\forall$-selectors, adding new variables $v'_\top$ and $v'_\bot$, and setting  
\[ \Psi:=\forall v'_\top, v'_\bot \exists v_\top, v_\bot \ E(v'_\top,v_\top) \wedge E(v_\top,v'_\top) \wedge E(v'_\bot,v_\bot) \wedge E(v_\bot,v'_\bot) \wedge \Psi'(v_\top,v_\bot).\] 
The purpose of universally quantifying the new variables $v'_\top$ and $v'_\bot$, instead of directly quantifying $v_\top$ and $v_\bot$, is to force $v'_\top$ and $v'_\bot$ to range over $\{1,3\}$ (recall that $E(v_\top,v_\top)$ and $E(v_\bot,v_\bot)$ are both atoms of $\Psi$). This is the same reason we add the vertex $\forall$ to the universal variable gadget.

We claim that $\mathcal{P}_{101} \models \Psi'(1,3)$ iff $\mathcal{P}_{101} \models \Psi$. It suffices to prove that $\mathcal{P}_{101} \models \Psi'(1,3)$ implies $\mathcal{P}_{101} \models \Psi'(3,1), \Psi'(1,1), \Psi'(3,3)$. The first of these follows by symmetry. The second two are easy to verify, and follow because the second literal in any clause is forbidden to be universally quantified in $\Phi$. If both paths $\top$ and $\bot$ are \mbox{w.l.o.g.} evaluated to $1$, then, even if some $l_1$- or $l_3$-literals are forced to evaluate to $3$, we can still extend this to a homomorphism from $\mathcal{G}_\Phi$ to $\mathcal{P}_{101}$.
\end{proof}
\noindent Note that the properties required for the final paragraph of the previous proof are inconsistent with our using the different diamond of \comments{Figure~\ref{fig:bad-gadget}}Figure~11 in our clause gadget (other than for the centre literal $l_2$).
\begin{figure}
\label{fig:bad-gadget}
\begin{center}
\includegraphics{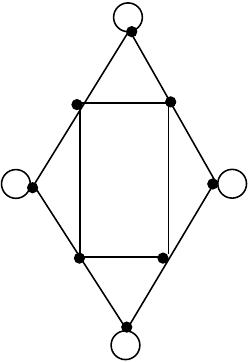}
\end{center}
\caption{Incorrect gadget for QCSP$(\mathcal{P}_{101})$.}
\end{figure}
\begin{proposition}
Let $\mathcal{P}_{0^a10^b10^c}$ be such that its centre is between its loops ($a+b \geq c$ and $b+c\geq a$). Then QCSP$(0^a10^b10^c)$ is Pspace-complete.
\end{proposition}
\begin{proof}
Let $m:=\max\{a,b,c\}$. The proof proceeds exactly as in Proposition~\ref{prop:P101}, except we use a new pattern and $\forall$-selector. We replace the pattern $\mathcal{P}_{101}$ with the pattern $\mathcal{P}_{10^b1}$. Note that the bracing of the clause diamonds is still on the first non-loop vertex. The $\forall$-selector $\mathcal{P}_{10}$ is replaced by $\mathcal{P}_{10^m}$. All new vertex-variables may be existentially quantified in the innermost block of quantifiers, except those on the path between the labelled $\forall$ vertex-variable in a universal variable gadget and that variable gadget's centre vertex-variable. These are quantified existentially immediately after the universal quantification of the labelled $\forall$ vertex-variable, in the obvious fashion. The new variable and clause gadgets for $\mathcal{P}_{00100010000}$ are drawn in \comments{Figure~\ref{fig:gadgets2}}Figure~12.
\begin{figure}
\label{fig:gadgets2}
\begin{center}
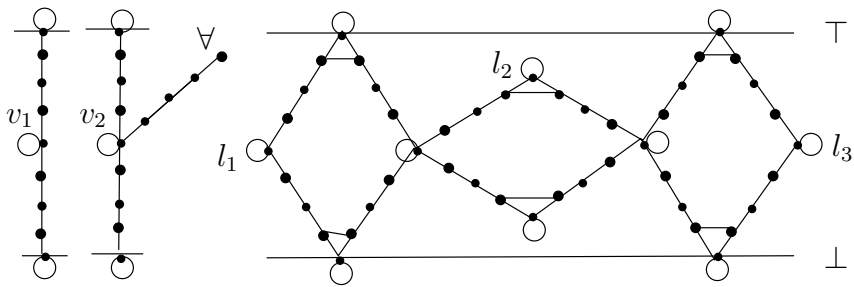
\end{center}
\caption{Variable and clause gadgets in reduction to QCSP$(\mathcal{P}_{00100010000})$; $b=3, m=4$.}
\end{figure}

We will give some explanation as to why this proof works. The crucial point is that there is 1.) an $m$-path from every vertex to one of the loops, 2a.) a point $p$ on the path such that there is an $m$-path from $p$ to the left loop and no $m$-path from $p$ to the right loop, and 2b.) a point $q$ on the path such that there is an $m$-path from $q$ to the right loop and no $m$-path from $p$ to the left loop. The conditions $a+b \geq c$ and $b+c\geq a$ are required to force properties 2a and 2b. The properties 1, 2a and 2b allow for a faithful simulation of QNAE3SAT.
\end{proof}
\begin{proposition}
\label{prop:weakly-balanced-0}
Let $\mathcal{P}$ be a weakly balanced $0$-centred path, then QCSP$(\mathcal{P})$ is Pspace-complete.
\end{proposition}
\begin{proof}
Let $\mathcal{P}_{0^a10^b10^c}$ be the path that is obtained from $\mathcal{P}$ if one retains only the loops nearest the centre, one on each side. Use the reduction of the previous proposition. The important point is that it is possible to evaluate $\{v_\top,v_\bot\}$ as the self-loops at $\{a+1,a+2+b\}$ -- and we will at some point be forced to choose, say, $v_\top$ to be $a+1$ and $v_\bot$ to be $a+2+b$ (when, e.g., $v'_\top$ is evaluated as $1$ and $v'_\bot$ is evaluated as $a+b+c+2$).
\end{proof}

\subsubsection{$\mathcal{P}_{10101}$ and weakly balanced $1$-centred paths}

We begin with the simplest weakly balanced $1$-centred path, $\mathcal{P}_{10101}$, which in some sense is also the trickiest.
\begin{figure}
\label{fig:gadgets10101A}
\begin{center}
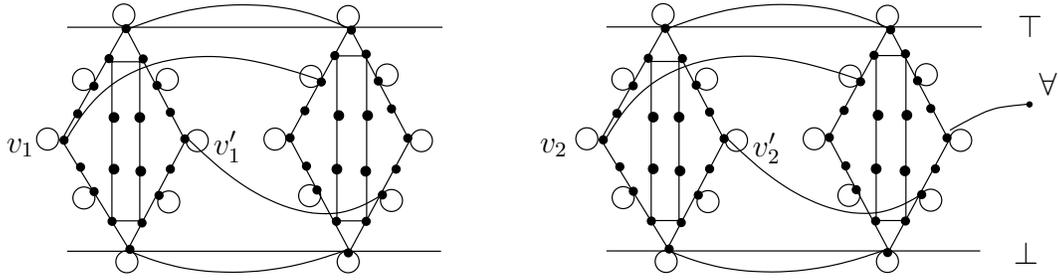
\end{center}
\caption{Variable gadgets in reduction to QCSP$(\mathcal{P}_{10101})$.}
\end{figure}
\begin{figure}
\label{fig:gadgets10101B}
\begin{center}
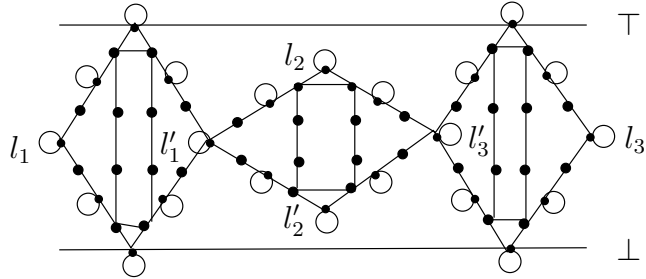
\end{center}
\caption{Clause gadget in reduction to QCSP$(\mathcal{P}_{10101})$.}
\end{figure}
\begin{proposition}
\label{prop:P10101}
QCSP$(\mathcal{P}_{10101})$ is Pspace-complete.
\end{proposition}
\begin{proof}
We work as in Proposition~\ref{prop:P101}, but with pattern $\mathcal{P}_{10101}$ and $\forall$-selector $\mathcal{P}_{10}$. We will need more sophisticated variable gadgets, along with some vertical bracing in the diamonds. The requisite gadgets are depicted in \comments{Figures~\ref{fig:gadgets10101A} and \ref{fig:gadgets10101A}}Figures 13 and 14. Finally, not only is $v_1$ (likewise, $v_2$) connected by an edge to a literal $l_i$ (if $v_1=l_i$), but on the other side $v'_1$ is also connected by an edge to $l'_i$. We assume for now that the paths $\top$ and $\bot$ are evaluated to $1$ and $5$. We need the extra edge from $l'_i$ to $v'_1$ as an evaluation of $l_1$ on a clause diamond to, e.g., $1$, no longer, in itself, enforces that $l'_i$ be evaluated to $5$. In the existential variable gadgets, $v_1$ must be evaluated to either $1$ or $5$, and $v'_1$ must be evaluated to the other. In a universal gadget, the loop adjacent to the vertex $\forall$ will be evaluated to any of $1$, $3$ or $5$ -- but $v_2$ and $v'_2$ must still be evaluated to opposites in $1$ and $5$. We depict an example of the situation where the loop adjacent to $\forall$ is evaluated to $3$, but the other vertices are mapped so as to set $v_2$ to $1$ and $v'_2$ to $5$ (this is the left-hand diamond of \comments{Figure~\ref{fig:135}}Figure~15).

Finally, we must explain what happens in the degenerate cases in which $v_\top$ and $v_\bot$ are not evaluated to $1$ and $5$, respectively (or vice-versa). It is not hard to see that this is no problem, even when universal variables are evaluated anywhere. Two examples of these degenerate cases, when $v_\top$ and $v_\bot$ are evaluated firstly to $1$ and $1$, and, secondly, to $1$ and $3$ are drawn in the centre and right of \comments{Figure~\ref{fig:135}}Figure~15. In both cases, we consider what happens when the evaluation of a universal variable forces the left-hand node of the gadget to be evaluated to $5$.
\end{proof}
\begin{figure}
\label{fig:135}
\begin{center}
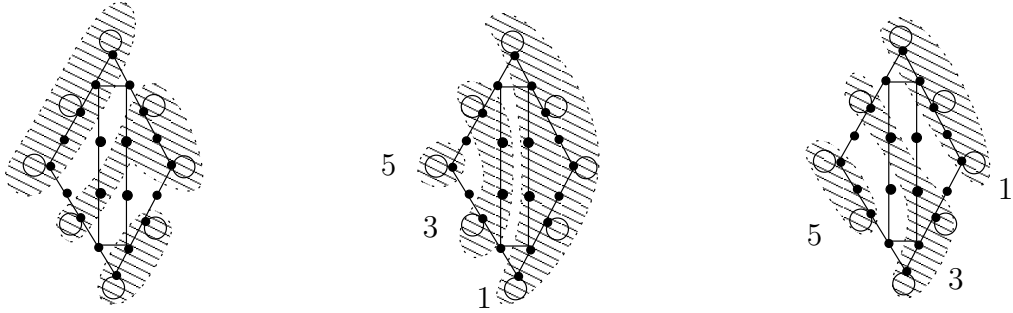
\end{center}
\caption{Degenerate mappings in QCSP$(\mathcal{P}_{10101})$.}
\end{figure}
\noindent It may be asked why we did not consider using a pattern of $\mathcal{P}_{101}$ and $\forall$-selector $\mathcal{P}_{10}$ in the previous proof, while, instead of beginning with $\forall v'_\top, v'_\bot$, using $\exists v'_\top \forall v'_\bot$. This would then select the centre loop for $v_\top$ along with at least once an outer loop for $v_\bot$. This proof would work for a simulation of the NP-hard NAE3SAT, but breaks down for the quantified variables of QNAE3SAT.

We will now briefly consider the paths $\mathcal{P}_{101^d01}$.
\begin{proposition}
For all $d$, QCSP$(\mathcal{P}_{101^d01})$ is Pspace-complete.
\end{proposition}
\begin{proof}
We may proceed as in the proof of Proposition~\ref{prop:P10101}, but with pattern $\mathcal{P}_{101^d01}$ and $\forall$-selector $\mathcal{P}_{1^d0}$. The vertical path braces in the diamonds become of length $d+2$, though one may verify that vertical bracing is actually unnecessary when $d \geq 2$. 
\end{proof}
\begin{proposition}
\label{prop:weakly-balanced-1}
If $\mathcal{P}$ is a weakly balanced $1$-centred path, then QCSP$(\mathcal{P})$ is Pspace-complete.
\end{proposition}
\begin{proof}
The weakly balanced $1$-centred cases have the form $\alpha 10^c1^d0^e1\beta$, in which $|\alpha|=a$, $|\beta|=b$, $c,d,e \neq 0$ and $a+c+2 < \frac{a+b+c+d+e+3}{2}< a+c+d+1$ (this final stipulation ensures that centre falls in the $1^d$).

If $a+c \neq e+b$, then let $m':=\min\{a+c,e+b\}$. If $a+c<e+b$ then set $l:=e$; if $e+b<a+c$ then set $l:=c$. We may proceed as in the proof of Proposition~\ref{prop:P101}, but with pattern $\mathcal{P}_{10^l1}$ and $\forall$-selector $\mathcal{P}_{1^d 0^{m'+1}}$.

Otherwise, $a+c = e+b$. Let $m:=\max\{c,e\}$.

If $a,b \leq \min\{c,e\}$ we may proceed as in the proof of Proposition~\ref{prop:P10101}, but with pattern $\mathcal{P}_{10^m1^d0^m1}$ and $\forall$-selector $\mathcal{P}_{1^d 0^{m}}$ (the vertical path braces in the diamonds should be of length $d+2m$).

If $c,e \leq \min\{a,b\}$ we may proceed as in the proof of Proposition~\ref{prop:P10101}, but with pattern $\mathcal{P}_{10^m1^d0^m1}$ and $\forall$-selector $\mathcal{P}_{1^d 0^{\max\{a,b\}}}$ (the vertical path braces in the diamonds should be of length $d+2m$).

If $a,e \leq \min\{b,c\}$ we may proceed as in the proof of Proposition~\ref{prop:P10101}, but with pattern $\mathcal{P}_{10^m1^d0^m1}$ and $\forall$-selector $\mathcal{P}_{1^d 0^{\max\{b,c\}}}$ (the vertical path braces in the diamonds should be of length $d+2m$).

If $b,c \leq \min\{a,e\}$ we may proceed as in the proof of Proposition~\ref{prop:P10101}, but with pattern $\mathcal{P}_{10^m1^d0^m1}$ and $\forall$-selector $\mathcal{P}_{1^d 0^{\max\{a,e\}}}$ (the vertical path braces in the diamonds should be of length $d+2m$).
\end{proof}

\subsubsection{Remaining path cases}

We are close to having exhausted the possible forms that a partially reflexive path may take.

\begin{proposition}
\label{prop:remaining-case}
Let $\mathcal{P}$ be of the form $\alpha 1^b 0^a$ such that $\mathcal{P}$ is not $0$-eccentric and $|\alpha|+1 \leq \frac{|\alpha|+b+a+1}{2} \leq |\alpha|+b$ (the centre is in the $1^b$ segment), then QCSP$(\mathcal{P})$ is Pspace-complete.
\end{proposition}
\begin{proof}
Note that $0 \leq a < |\alpha|$. Since $\mathcal{P}$ is of the form $\gamma:=\alpha 1^b 0^a$ and not $0$-eccentric, there exists a right-most $10 \in \alpha$ at position $c\geq a$ (this is the position of the $1$). That last $10$ is of the form $10^e1$ and then it hits a sequence of $1$s that merge into the $1^b$ in the centre of $\gamma$. We may proceed as in the proof of Proposition~\ref{prop:P101}, but with pattern $\mathcal{P}_{10^e1}$ and $\forall$-selector $\mathcal{P}_{1^{b+a-c} 0^c}$.
\end{proof}

\begin{theorem}
If $\mathcal{P}$ is not a $0$-eccentric path, then QCSP$(\mathcal{P})$ is Pspace-complete.
\label{thm:paths-hard}
\end{theorem}
\begin{proof}
Suppose $\mathcal{P}$ is not a $0$-eccentric path. Then,
if $\mathcal{P}$ is weakly balanced, the result follows from Propositions~\ref{prop:weakly-balanced-0} and \ref{prop:weakly-balanced-1}. Otherwise, $\mathcal{P}$ is of the form of Proposition~\ref{prop:remaining-case}, and the result follows from that proposition.
\end{proof}

\subsection{NP-hardness for remaining trees}

\begin{theorem}
\label{thm:not-quasi-loop-connected}
Let $\mathcal{T}$ be a tree that is not quasi-loop-connected. Then QCSP$(\mathcal{T})$ is NP-hard.
\end{theorem}
\begin{proof}
Let $\mathcal{T}$ and its associated $\lambda:=\lambda_T$ be given. Define $\mu(x,y)$ to be the minimum distance between some reflexive subtree $\mathcal{T}_x$ (at distance $\lambda$ from $x$) and some reflexive subtree $\mathcal{T}_y$ (at distance $\lambda$ from $y$). Note that we are considering all possible reflexive subtrees $\mathcal{T}_x$ and $\mathcal{T}_y$. In particular, since $\mu(x,y)$ is a minimum, it is sufficient to consider only such reflexive subtress that are maximal under inclusion. Let $\mu:=\max\{\mu(x,y):x,y \in T\}$. Since $\mathcal{T}$ is not quasi-loop connected, $\mu>1$. A subpath $\mathcal{P} \subseteq \mathcal{T}$ is said to have the \emph{$\mu$-property} if it connects two (maximal under inclusion) reflexive subtrees $\mathcal{T}_x$ and $\mathcal{T}_y$ that witness the maximality of $\mu$, as just defined. Let $\nu$ be the size of the largest induced reflexive subtree of $\mathcal{T}$.

Let $\mathbb{P}$ be the set of induced subpaths $\mathcal{P}$ of $\mathcal{T}$ that have the $\mu$-property, relabelled with vertices $\{1,\ldots,n:=|P|\}$ in the direction from $\mathcal{T}_x$ to $\mathcal{T}_y$. Note that the paths in $\mathcal{P}$ have loops on neither vertex $2$ nor vertex $n-1$. Note also that $\mathbb{P}$ is closed under reflection of paths (i.e., the respective mapping of $1,\ldots,n$ to $n,\ldots,1$).
We would like to reduce from NAE3SAT exactly as in the proof of Proposition~\ref{prop:P101}, with pattern $\mathcal{P}_{10^{\mu-1}1}$ and $\forall$-selector $\mathcal{P}_{1^\nu 0^\lambda}$. The sentence we would create for input for QCSP$(\mathcal{T})$ has precisely two universal quantifiers, at the beginning (i.e. this is the only use of the $\forall$-selector). The point is that somewhere we would forcibly stretch $v_\top$ and $v_\bot$ to be at distance $\mu$ (when this distance is less, it will only make it easier to extend to homomorphism). However, this method will only succeed if there is the path $\mathcal{P}_{10^{\mu-1}1} \in \mathbb{P}$.

For $\mathcal{P} \in \mathbb{P}$, let $\Delta(\mathcal{P})$ be the distance from the end of the path $\mathcal{P}$ (vertex $n$) to the nearest loop. Let $\Delta:=\max\{\Delta(\mathcal{P}):\mathcal{P} \in \mathbb{P}\}$. We build an input $\Psi$ for QCSP$(\mathcal{T})$ as in the proof of Proposition~\ref{prop:P101}, with pattern $\mathcal{P}_{10^{\Delta-1}1}$, except for the point at which we have only the variables $v_\top$ and $v_\bot$ remaining free (i.e., the one place we would come to use a $\forall$-selector).  Here, we use the $\forall$-selector $\mathcal{P}_{0^{\mu-1} 1^\nu 0^\lambda}$ for $v_\top$ and $\mathcal{P}_{1^\nu 0^\lambda}$ for $v_\top$. For the correctness of this, note that a walk of $\mu-1$ will always get you to the penultimate loop along a path $\mathcal{P} \in \mathbb{P}$, which is sometimes at distance $\Delta$ from the end (and is always at distance $\leq \Delta$ from the end).
\end{proof}

\section*{Acknowledgements}

The author is grateful for assistance with majority operations from Tom\'as Feder and Andrei Krokhin.

\bibliographystyle{acm}
\bibliography{local}

\begin{thebibliography}{10}

\bibitem{Bandelt1987191}
{\sc Bandelt, H.~J., Dählmann, A., and Schütte, H.}
\newblock Absolute retracts of bipartite graphs.
\newblock {\em Discrete Applied Mathematics 16}, 3 (1987), 191 -- 215.

\bibitem{Bandelt:1989:DAR:72175.72177}
{\sc Bandelt, H.-J., and Pesch, E.}
\newblock Dismantling absolute retracts of reflexive graphs.
\newblock {\em Eur. J. Comb. 10\/} (May 1989), 211--220.

\bibitem{barto:1782}
{\sc Barto, L., Kozik, M., and Niven, T.}
\newblock The {CSP} dichotomy holds for digraphs with no sources and no sinks
  (a positive answer to a conjecture of {Bang-Jensen} and {Hell}).
\newblock {\em SIAM Journal on Computing 38}, 5 (2009), 1782--1802.

\bibitem{BBCJK}
{\sc B{\"o}rner, F., Bulatov, A.~A., Chen, H., Jeavons, P., and Krokhin, A.~A.}
\newblock The complexity of constraint satisfaction games and qcsp.
\newblock {\em Inf. Comput. 207}, 9 (2009), 923--944.

\bibitem{OxfordQuantifiedConstraints}
{\sc B\"orner, F., Krokhin, A., Bulatov, A., and Jeavons, P.}
\newblock Quantified constraints and surjective polymorphisms.
\newblock Tech. Rep. PRG-RR-02-11, Oxford University, 2002.

\bibitem{Bulatov}
{\sc Bulatov, A.}
\newblock A dichotomy theorem for constraint satisfaction problems on a
  3-element set.
\newblock {\em J. ACM 53}, 1 (2006), 66--120.

\bibitem{JBK}
{\sc Bulatov, A., Krokhin, A., and Jeavons, P.~G.}
\newblock Classifying the complexity of constraints using finite algebras.
\newblock {\em SIAM Journal on Computing 34\/} (2005), 720--742.

\bibitem{hubie-sicomp}
{\sc Chen, H.}
\newblock The complexity of quantified constraint satisfaction: Collapsibility,
  sink algebras, and the three-element case.
\newblock {\em SIAM J. Comput. 37}, 5 (2008), 1674--1701.

\bibitem{LICS2008}
{\sc Chen, H., Madelaine, F., and Martin, B.}
\newblock Quantified constraints and containment problems.
\newblock In {\em 23rd Annual IEEE Symposium on Logic in Computer Science\/}
  (2008), pp.~317--328.

\bibitem{DalmauK08}
{\sc Dalmau, V., and Krokhin, A.~A.}
\newblock Majority constraints have bounded pathwidth duality.
\newblock {\em Eur. J. Comb. 29}, 4 (2008), 821--837.

\bibitem{FederVardi}
{\sc Feder, T., and Vardi, M.}
\newblock The computational structure of monotone monadic {SNP} and constraint
  satisfaction: {A} study through {D}atalog and group theory.
\newblock {\em {SIAM} Journal on Computing 28\/} (1999), 57--104.

\bibitem{GolovachPaulusmaSong}
{\sc Golovach, P., Paulusma, D., and Song, J.}
\newblock Computing vertex-surjective homomorphisms to partially reflexive
  trees.
\newblock In {\em Proceedings of the 6th International Computer Science
  Symposium in Russia (CSR 2011), St. Petersburg, Russia, June 14-18, 2011,
  LNCS (to appear)\/} (2011).

\bibitem{HellNesetril}
{\sc Hell, P., and Ne\v{s}et\v{r}il, J.}
\newblock On the complexity of {H}-coloring.
\newblock {\em Journal of Combinatorial Theory, Series B 48\/} (1990), 92--110.

\bibitem{CiE2006}
{\sc Martin, B., and Madelaine, F.}
\newblock Towards a trichotomy for quantified {H}-coloring.
\newblock In {\em 2nd Conf. on Computatibility in Europe, LNCS 3988\/} (2006),
  pp.~342--352.

\bibitem{MatchingCut}
{\sc Patrignani, M., and Pizzonia, M.}
\newblock The complexity of the matching-cut problem.
\newblock In {\em Graph-Theoretic Concepts in Computer Science, 27th
  International Workshop, WG 2001\/} (2001), pp.~284--295.

\bibitem{Schaefer}
{\sc Schaefer, T.~J.}
\newblock The complexity of satisfiability problems.
\newblock In {\em Proceedings of STOC'78\/} (1978), pp.~216--226.

\end{thebibliography}

\end{document}